\documentclass[10pt,twoside,leqno]{amsart}  
\usepackage[utf8]{inputenc} 
\usepackage{amsmath}
\usepackage{amsfonts}
\usepackage{amsthm}
\usepackage{amssymb}
\usepackage{url}
\usepackage{algorithmic}
\usepackage{algorithm}

\usepackage{graphicx} 
\usepackage{color}
\usepackage{array} 
\usepackage{fancyhdr} 
\usepackage[T1]{fontenc}
\usepackage{utopia}

\newcommand{\ZZ}{\mathbb{Z}}
\newcommand{\QQ}{\mathbb{Q}}
\newcommand{\NN}{\mathbb{N}}
\newcommand{\CC}{\mathbb{C}}
\newcommand{\RR}{\mathbb{R}}

\theoremstyle{plain}

\newtheorem{lem}{Lemma}

\theoremstyle{definition}

\theoremstyle{remark}
\newtheorem{rem}{Remark}

\begin{document}
\title{How to compute the constant term of a power of a Laurent polynomial efficiently}
\author{Pavel Metelitsyn}
\begin{abstract}
We present an algorithm for efficient computation of the constant term of a power of a multivariate Laurent polynomial. The algorithm is based on univariate interpolation, does not require the storage of intermediate data and can be easily parallelized.
As an application we compute the power series expansion of the principal period of some toric Calabi-Yau varieties and find previously unknown differential operators of Calabi-Yau type.\end{abstract}
\maketitle

\section{Introduction}
The constant terms of powers of Laurent polynomials appear in different areas of mathematics and theoretical physics \cite{D}, \cite{DvK}. They show up prominently in the expansion of the {\em fundamental period} 
\begin{equation}\nonumber
\Phi(z)=\frac{1}{(2\pi i)^N}\int_{\Gamma} \frac{1}{1-zf(X)}\frac{dX_1}{X_1}\frac{dX_2}{X_2}\ldots \frac{dX_N}{X_N}\label{period-integral},
\end{equation} 
of a Laurent polynomial in $X_1,\dots,X_N$, and $\Gamma$ is the $N$-dimensional
torus cycle $|X_i|=\epsilon$. In the theory of mirror symmetry for
Calabi-Yau hypersurfaces in a toric variety, this fundamental period is of
great importance, \cite{BK}, \cite{BvS}.
To compute the Taylor expansion of $\Phi$ in the neighborhood of $z=0$ we write 
\begin{equation}\nonumber
\frac{1}{1-zf(X)}=\sum_{i\geq 0} (f(X))^iz^i
\end{equation} and integrate using Cauchy's formula. The $n$-th coefficient in the power series expansion of $\Phi$ is therefore the constant term of the $n$-th power of the Laurent polynomial $f$. 
For example the constant term of the 151st power of 
\begin{eqnarray}\label{big-poly-ex}\nonumber
f &=& \frac{X}{YZ}+\frac{ZT}{XY}+ZT+T+\frac{Y}{X}+\frac{Z}{X}+\frac{Y}{XT}+Y+\frac{XY}{ZT}+\frac{X}{Z}+\frac{Y}{T}\\ 
  &+& Z+\frac{1}{T}+X+\frac{X}{Y}+\frac{ZT}{Y}+\frac{X}{ZT}+\frac{Y}{ZT}+\frac{1}{X}+\frac{1}{Y}+\frac{ZT}{X}+\frac{T}{Y}+\frac{1}{Z}\nonumber
\end{eqnarray}
is\\
$15412036066982883611159466717890839926274227993361685769096965357956125\\
0836097113850549748895583119569242295079022614473032754474202469738117581\\
03097074502829198076370950235391810731785760778732696320,$ a number with 200 digits.

In this paper we consider the following general problem:\\ 
{\em How to compute
for a multivariate (Laurent) polynomial $f$ the coefficient of the monomial $X^\alpha, \, \alpha=(\alpha_1,\dots,\alpha_n)$ of the $p$-th power, $p\in\mathbb{N}$ of $f$.}\\

At first glance this problem is trivial. Why cant't we just split the problem in smaller ones and compute $f^{r}=:g=\sum b_\beta X^\beta$ and $f^{s}=:h=\sum c_\beta X^\beta$ s.t. $r+s=p$ and get the coefficient we are interested in by Cauchys product formula 
\begin{eqnarray}
\text{coeff. of } X^\alpha \text{ in } f^p=\sum_{i+j=\alpha} b_ic_j .\label{cauchy-naive}
\end{eqnarray}
A general polynomial in $n$ variables in which the degree of each individual variable is not greater than $d$ has $(d+1)^n$ monomials. In our example $n=4, \, p=151, \, d=2$ after multiplication with the common denominator $XYZT$. In order to use (\ref{cauchy-naive}) we could compute $f^{75}$ and $f^{76}$ which will require to store about $2(2\cdot 76+1)^4\approx10^9$ coefficients of intermediate polynomials $g$ and $h$. While this number is not really big, we must keep in mind that, in general, the absolute value of the coefficients grow exponentially as $p$ increases. Therefore, the amount of memory needed to store one particular coefficient will, typically, linearly grow with $p$. Another possibility is to use a multivariate discrete Fourier transform: evaluate the polynomial $f$ on sufficently many points of the $n$-dimensional grid of roots of unity of sufficient high degree, compute the $p$-th power of the values and apply the inverse discrete Fourier transform to get the single coefficient we are interested in. Here, again, the number of interpolation nodes and therefore of the intermediate values which have to be stored is of the same order as the number of monomials in $f^p$ i.e. $(2\cdot 151)^4$. 

In this paper, we present an efficient and, we believe, simple algorithm which adresses this problem. Our approach is similar to the DFT method in that it based on evaluation and iterpolation. However, the amount of memory needed is much lower and the running speed can be improved by a trivial parallelization of the evaluation. We also describe the implementation of the parallel version of the algorithm for the CUDA-enabled graphics hardware. Finally, an application to the problem of finding the Picard-Fuchs differential equation of the family of toric Calabi-Yau varieties is given. This problem was the main motivation for this work.\\
\vskip 10pt
{\bf Acknowledgement:} I like to thank D. van Straten for suggesting the
problem and E. Sch\"omer for useful discussions. This work was realised with 
support from the DFG priority program
{\em Algorithmic and Experimental Methods in Algebra, Geometry and Number 
Theory}, as part of the project {\em Monodromy algorithms in {\sc Singular}}.

%
%
%
%
\section{Algorithm}
\subsection{Notations and Preliminaries}
 Let $f$ be a multivariate polynomial in variables $X_1,\dots,X_k$ with rational coefficients. Using multiindex notation we write
\[f=\sum_{i\in\Delta}c_iX^i,\]
$i:=(i_1,\dots,i_k)$ and $\Delta\subset\ZZ^k$ a finite index set.
Let $\deg_r(f)$ denote the degree of $f$ in the variable $X_r$. 
For a given $k$-dimensional multiindex $j$ we denote by $\left[f\right]_j$ the coefficient of $f$ in front of $X^j$ i.e.
\[\left[f\right]_j:=c_j.\]

 Furthermore let $j'$ denote the $(k-1)$-dimensional multindex obtained by truncating $j$ after the $(k-1)$-th component and for an $r$-dimensional multiindex $i$ let $ij$ denote the concatenation of both i.e. for $j=(j_1,\dots,j_k)$ and $i=(i_1,\dots,i_r)$
\begin{eqnarray}
j':&=&(j_1,\dots,j_{k-1}),\nonumber\\
ij:&=&(i_1,\dots,i_r,j_1,\dots,j_k).\nonumber
\end{eqnarray}

Given $N+1$ pairs $(u_{\alpha},w_{\alpha})\in\mathbb{Q}^{2},\, \alpha=0,\dots,N$
such that $u_{\alpha}\neq u_{\beta}$ for $\alpha\neq \beta$, there is a unique \emph{univariate interpolation polynomial} of degree $N$ such that
\[\mathcal{I}_N(u_{\alpha})=w_{\alpha},\,\forall\alpha\in\left\{0,\dots,N\right\}\]
\[\mathcal{I}_{N}(u_0,\dots,u_N;w_0,\dots,w_N)(Y)=\sum_{i=0}^N w_i l_i(Y),\]
where 
\begin{eqnarray*}
l_i(Y) &=& \prod_{i\neq j}\frac{Y-u_i}{u_j-u_i}
\end{eqnarray*}
is the Lagrange basis polynomial for $u_i$. $\mathcal{I}_N$ has rational coeffiients and $\mathcal{I}_N(u_{\alpha})=w_{\alpha},\,\forall\alpha=0,\dots,N$. 

Let $f_u$ denote the polynomial $f$ with $u$ substituted for the last variable i.e.

\[f_u(X_1,\dots,X_{n-1})=f(X_1,\dots,X_{n-1},u).\]
We have the following observation about the coefficients of the interpolation polynomial and those of $f_u$.
\begin{lem}\label{lem1} Let $f\in \QQ[X_1,\dots,X_n]$, $N:=\deg_{n}(f)$, $\left\{(u_\alpha,w_\alpha)\right\}_{\alpha=0,\dots,N}$ as above, and $i=(i_{1}, \dots,i_{n})\in\mathbb{N}^n$. Then, we have 
\begin{equation}
[f]_{i}=\left[\mathcal{I}^{N}(u_{0},\dots,u_{N};[f_{u_0}]_{i'},\dots,[f_{u_N}]_{i'})\right]_{i_{n}}\label{eq:1}.
\end{equation}
\end{lem}
\begin{proof} Let 
\begin{equation}
f(X_{1},\dots,X_{n})=\sum_{\nu_{1},\nu_{2},\dots,\nu_{n}}c_{\nu_{1}\nu_{2}\dots\nu_{n}}X_{1}^{\nu_{1}}X_{2}^{\nu_{2}}\dots X_{n}^{\nu_{n}},\label{polynom}\end{equation}
and \begin{equation}
\mathcal{I}^{N}(u_{0},\dots,u_{N};[f_{u_0}]_{i'},\dots,[f_{u_N}]_{i'})(Y)=\sum_{k=0}^{N}a_{k}Y^{k}.\label{eq:interpol}\end{equation}
Then, the right hand side of (\ref{eq:1}) equals $a_{i_{n}}$ and, therefore, we have to show that $[f]_{i}=a_{i_{n}}$.
The coefficients $a_{0},\dots,a_{n}$ of the interpolation polynomial are found by multiplying the vector $([f_{u_0}]_{i'},[f_{u_1}]_{i'},\dots,[f_{u_N}]_{i'})^{T}$
by the inverse of the Vandermonde matrix for $u_{0},\dots,u_{N}$. This inverse exists since $u_{0},\dots,u_{N}$ are pairwise distinct.
Therefore, we have \begin{equation}
\left(\begin{array}{c}
a_{0}\\
a_{1}\\
\vdots\\
a_{N}\end{array}\right)=\left(\begin{array}{ccccc}
1 & u_{0} & u_{0}^{2} & \cdots & u_{0}^{N}\\
1 & u_{1} & u_{1}^{2} & \cdots & u_{1}^{N}\\
\vdots & \vdots & \vdots & \ddots & \vdots\\
1 & u_{N} & u_{N}^{2} & \cdots & u_{N}^{N}\end{array}\right)^{-1}\left(\begin{array}{c}
\left[f_{u_0}\right]_{i'}\\
\left[f_{u_1}\right]_{i'}\\
\vdots\\
\left[f_{u_N}\right]_{i'}\end{array}\right)\label{zurueck}\end{equation}
By (\ref{polynom}) we have
\begin{eqnarray*}
[f_{u_j}]_{i'} & = & \left[\sum_{\nu_{1},\nu_{2},\dots,\nu_{n-1},\tau}c_{\nu_{1}\nu_{2}\dots\tau}u_{j}^{\tau}X_{1}^{\nu_{1}}X_{2}^{\nu_{2}}\dots X_{n-1}^{\nu_{n-1}}\right]_{i'}\\
 & = & \sum_{\tau}c_{i_{1}i_{2}\dots i_{n-1}\tau}u_{j}^{\tau}\\
 & = & \sum_{\tau}c_{i'\tau}u_{i}^{\tau},\end{eqnarray*}
where in the last expression $i'\tau$ is an abbriviration for the
multiindex $(i_{1},\dots,i_{n-1},\tau$). Using this, rewrite (\ref{zurueck}) as
\begin{equation}
\left(\begin{array}{c}
a_{0}\\
a_{1}\\
\vdots\\
a_{N}\end{array}\right)=\left(\begin{array}{ccccc}
1 & u_{0} & u_{0}^{2} & \cdots & u_{0}^{N}\\
1 & u_{1} & u_{1}^{2} & \cdots & u_{1}^{N}\\
\vdots & \vdots & \vdots & \ddots & \vdots\\
1 & u_{N} & u_{N}^{2} & \cdots & u_{N}^{N}\end{array}\right)^{-1}\left(\begin{array}{c}
\sum_{\tau}c_{i'\tau}u_{0}^{\tau}\\
\sum_{\tau}c_{i'\tau}u_{1}^{\tau}\\
\vdots\\
\sum_{\tau}c_{i'\tau}u_{N}^{\tau}\end{array}\right)
\label{zurueck-again}
\end{equation}
At the same time 
\begin{equation}
\left(\begin{array}{c}
\sum_{\tau}c_{i'\tau}u_{0}^{\tau}\\
\sum_{\tau}c_{i'\tau}u_{1}^{\tau}\\
\vdots\\
\sum_{\tau}c_{i'\tau}u_{N}^{\tau}\end{array}\right)=\left(\begin{array}{ccccc}
1 & u_{0} & u_{0}^{2} & \cdots & u_{0}^{N}\\
1 & u_{1} & u_{1}^{2} & \cdots & u_{1}^{N}\\
\vdots & \vdots & \vdots & \ddots & \vdots\\
1 & u_{N} & u_{N}^{2} & \cdots & u_{N}^{N}\end{array}\right)\left(\begin{array}{c}
c_{i'0}\\
c_{i'1}\\
\vdots\\
c_{i'N}\end{array}\right).\label{hin}\end{equation}
 From (\ref{hin}) and (\ref{zurueck-again}) we get \[
\left(\begin{array}{c}
a_{0}\\
a_{1}\\
\vdots\\
a_{N}\end{array}\right)=\left(\begin{array}{c}
c_{i'0}\\
c_{i'1}\\
\vdots\\
c_{i'N}\end{array}\right)\]
 and with (\ref{polynom})
 \[a_{i_{n}}=c_{i'i_{n}}=c_{i}=[f]_{i}\]
follows.
\end{proof}
\noindent
Since $\deg_{n}(f^p)=p\cdot\deg_{n}(f)$ we have
\begin{equation}\label{cor_1}
\left[f^p\right]_{i}=\left[\mathcal{I}^{N}(u_{0},\dots,u_{N};[f_{u_0}^p]_{i'},\dots,[f^p_{u_N}]_{i'})\right]_{i_{n}},
\end{equation}
where $N=p\cdot\deg_{n}(f)$.

\subsection{The algorithm and its application to the constant term series of a Laurent polynomial } 
Let $h\in\mathbb{Q}[X_1^\pm,\dots,X_n^\pm]$ be a Laurent polynomial. We want to calculate the first coefficients of its constant terms series $([h^p]_0)_{p\in\NN}$. Let $S=X_1^{s_1}\dots X_n^{s_k},\ s_i\geq 0,$ be the common denominator of the Laurent monomials of $h$. Then $Sh\in\QQ[X_1,\dots,X_k]$ and
 \begin{equation}\label{power_of_laurent}\nonumber
[h^p]_0=[S^pf^p]_{p\cdot s},
\end{equation}
where $p\cdot s:=(ps_1,\dots,ps_n)$. We want to apply (\ref{cor_1}) to $f=Sh$ and $i=p\cdot s$. 

The simple idea now is to compute $\left[f_{u_j}^p\right]_{i'},\, j=0,\dots,N$ which appear in (\ref{cor_1}) recursively i.e.
\begin{equation}
\left[f_{u_j}^p\right]_{i'}=\left[\mathcal{I}^{N}(u_{0},\dots,u_{N};[f_{u_j u_0}^p]_{i''},\dots,[f^p_{u_j u_N}]_{i''})\right]_{i_{n-1}}
\end{equation}
and so on for $f_{u_ju_i},f_{u_ju_iu_k},\dots$. 
With each level of recursion the number of variables decreases by one, therefore the depth of the recursion equals $n$.
For simplicity assume 
\begin{equation*}\label{asump}
\deg_{1}(f)=\dots=\deg_{n}(f)=d.
\end{equation*}
The computation of $\left[f^p\right]_i$ proceeds as follows:
\begin{enumerate}
\item For $N=dp$ fix $u_0,\dots,u_{N}$ as in Lemma and compute $V^{-1}=(V_{ij})$, where $V$ is the Vandermonde matrix of $(u_0,\dots,u_{N})$. \item Invoke the recursive procedure $\operatorname{COEFF}(i,p,k,A)$ with initial parameters $k=n$ and $A=$''coefficient matrix of $f$''.
\end{enumerate}
The pseudo code of the procedure $\operatorname{COEFF}$ is given in the Algorithm 1.
\renewcommand{\algorithmicrequire}{\textbf{Input:}}
\renewcommand{\algorithmicensure}{\textbf{Output:}}
\renewcommand{\algorithmiccomment}[1]{// #1}
\noindent
\begin{algorithm}
\caption{$\operatorname{COEFF}(k,i,A,p)$}
\label{alg-coeff}
\begin{algorithmic}[1]
\REQUIRE $k>0$, $i\in\mathbb{N}^k$, $A\in Mat(\underbrace{d\times\cdots\times d}_{k})$,  $p\geq 1$
\ENSURE $\left[f^p\right]_i$
\FORALL{$0\leq s\leq N$}
\STATE $B_{i_1,\dots,i_{k-1}}\leftarrow\sum_j A_{i_1,\dots,i_{k-1},j}u_s^j$\COMMENT Do it using Horner's rule
\IF{$k=1$}
\STATE $w_s\leftarrow B^p$ \COMMENT{$B$ has $k-1=0$ indices i.e. $B$ is a scalar}  
\ELSE
\STATE $w_s\leftarrow \operatorname{COEFF}(k-1,i',B,p)$
\ENDIF
\ENDFOR
\STATE $S \leftarrow \sum_{s=0}^N w_sV_{i_{k}s}$ \COMMENT Compute the $i_k$-th component of $V^{-1}w$
\RETURN $S$
\end{algorithmic}
\end{algorithm}

To improve the performane we apply some well known computational tricks. The evaluation in the line 2 is done using the Horner's rule which for a polynomial $h(x)=a_0+a_1x+\dots+a_lx^l$ at $x=x_0$ says
\[f(x_0)=a_0+x_0(a_1+x_0(a_2+x_0(\dots))).\]
In this way the expensive computation of powers of $x_0$ is avoided.
The Vandermonde matrix $V$ can be inverted using classical Gauss elimination algorithm in $O(N^3)$ time. Since on each level of recursion we need only one row of $V^{-1}$ (line 9 in the code above) precomputing the whole inverse seems too expensive. Fortunately, if the nodes $u_0,\dots,u_N$ are choosen to be equidistant there is a recursive algorithm which computes one single row of $V^{-1}$ in $O(N^2)$ \cite{turner}. The computation of the $p$-th power in the line 4 can be done using binary powering algorithm as described in \cite{knuth}. We also note that $w_s$'s in line 6 do not depend on each other and thus the order in which they are computed does not matter. In particular they can be computed parallely. Finally, the number of nodes $N$ need not to be the same on each recursion level but depend on the degree of $f$ in a particular variable i.e. we have to choose $N_k:=\max\left\{\deg_{1}(f),\dots,\deg_{k}(f)\right\}\cdot p$ on the $(n-k)$-th recursion level. For example if $f$ has degree $2$ in one variable and degree $1$ in the remaining variables we could improve speed by considering $2p$ interpolation nodes only on one level of recursion and $p$ interpolation nodes on the ramaining.

If $f$ happens to have special form 
\begin{equation}\label{decomp}
f=A+BX_1+CX^2_{1}
\end{equation} 
with $A,B,C\in\mathbb{Z}[X_2,\dots,X_k]$ a trick can be applied to slightly speed up the computation. If (\ref{decomp}) holds, then so is for $f_u$,
\begin{equation}\nonumber
f_u=A_u+B_uX_1+C_uX^2_{1}
\end{equation}
with $A_u,B_u,C_u\in\mathbb{Z}[X_2,\dots,X_{k-1}]$. Consequently, this decomposition holds on all levels of recursion. We note that in this case
\begin{equation}\nonumber
f^{p} = \sum_{i,j}{p \choose i,j}A^{i}C^{j}B^{p-i-j}X_1^{p-i+j},
\end{equation}
and
\begin{eqnarray}\label{split2}
\left[f^{p}\right]_p & = &\left[\sum_{i=0}^{p}{p \choose i,i}A^{i}C^{i}B^{p-2i}\right]_{p}\label{split}\\
 & = & \sum_{i=0}^{p}{p \choose i,i}\left[A^{i}C^{i}B^{p-2i}\right]_{p}\nonumber.
\end{eqnarray}
since we need only the terms with $X_1^{p-i+j}=X_1^p$ i.e. $i=j$. We apply this ''shortcut'' on the last but one level of recursion to compute $\left[f^p_{u_1,\dots,u_{k-2}}\right]_p$ when we have to deal with polynomials in two variables. We use the modified version of COEFF and a new procedure SPLIT2 for this. Thus the depth of recursion is reduced by one. For the case $n=4$ the call tree of the algorithm is depicted in Figure \ref{tree}. 

\noindent
\begin{algorithm}
\caption{$\operatorname{COEFF}(k,i,A,p)$}
\label{alg-coeff-mod}
\begin{algorithmic}[1]
\REQUIRE $k>0$, $i\in\mathbb{N}^k$, $A\in Mat(\underbrace{d\times\cdots\times d}_{k})$,  $p\geq 1$
\ENSURE $\left[f^p\right]_i$
\FORALL{$0\leq s\leq N$}
\STATE $B_{i_1,\dots,i_{k-1}}\leftarrow\sum_j A_{i_1,\dots,i_{k-1},j}u_s^j$
\IF{$k=2$}
\STATE $w_s\leftarrow \operatorname{SPLIT2}(i',B,p)$
\ELSE
\STATE $w_s\leftarrow \operatorname{COEFF}(k-1,i',B,p)$
\ENDIF
\ENDFOR
\STATE $S \leftarrow \sum_{s=0}^N w_sV_{is}$
\RETURN $S$
\end{algorithmic}
\caption{Modified version of COEFF}
\end{algorithm}

\begin{algorithm}
\caption{$\operatorname{SPLIT2}(i,A,p)$}
\label{alg-split2}
\begin{algorithmic}[1]
\REQUIRE $k\geq 0$, $i\in\mathbb{N}^k$, $A\in Mat(\underbrace{d\times\cdots\times d}_{k})$,  $p\geq 1$
\ENSURE $\left[f^p\right]_i$
\STATE $s \leftarrow 0$
\STATE $m \leftarrow \frac{N}{2}+1$
\FORALL{$0\leq i\leq N$}
\STATE $w_{i}\leftarrow (A_{00}+u_i(A_{01}+u_iA_{02}))\cdot(A_{20}+u_i(A_{21}+u_iA_{22}))$ 
\STATE $t\leftarrow A_{10}+u_i(A_{11}+u_iA_{12})$
\IF{$n\equiv 0\mod 2$}
\STATE $v_{m-1,i}\leftarrow 1$
\ELSE
\STATE $v_{m-1,i}\leftarrow t$
\ENDIF
\STATE $t\leftarrow t^2$
\FORALL{$0\leq j\leq m-2$}
\STATE $v_{m-2-j,i}\leftarrow t\cdot v_{m-1-j,i}$
\ENDFOR
\STATE $w'_i\leftarrow 1$
\STATE $z_i\leftarrow v_{0,i}$
\ENDFOR
\STATE $t\leftarrow 0$
\FORALL{$0\leq i\leq N$}
\STATE $s\leftarrow s + c_iz_i$
\ENDFOR
\STATE $s\leftarrow M_0\cdot s$
\FORALL{$1\leq j\leq m$}
\STATE $t\leftarrow 0$
\FORALL{$0\leq i\leq N$}
\STATE $w'_i\leftarrow w'_i\cdot w_i$
\STATE $t\leftarrow t+w'_i\cdot w_i\cdot c_i\cdot v_{j,i}$
\ENDFOR
\STATE $s\leftarrow s+M_j\cdot t$
\ENDFOR
\RETURN $s$
\end{algorithmic}
\end{algorithm}

The bottom procedure SPLIT2 is called exactely $N^{n-2}$ times. Inside the SPLIT2 procedure we have two nested loops of depth 2. The bodies of both inner loops (lines 12-14 and 25-28 in the listing \ref{alg-split2}) are executed both at most $\frac{N^2}{2}$ times. Therefore the time complexity of the algorithm is of order $O(N^{n})$. Thanks to SPLIT2 we gain a factor $\frac{1}{2^{n-2}}$ compared to the version which does not use SPLIT2. If we take into consideration the linear growth of the length of the numbers involved into computation we arrive at assymtotical running time $pN^n$ or $d^np^{n+1}$, since $N$ depends linearly on $p$ and $d$.

As we can see from the description of the algorithm, the memory consumption is $O(N)=O(dp)$, which is the space needed to store the $n-1$ rows of the inverse Vandermonde matrix. In particular, for fixed $d$ it is independent from the number of monomials $f$ has.

\begin{figure}\label{tree}
\begin{center}
\includegraphics[width=10cm]{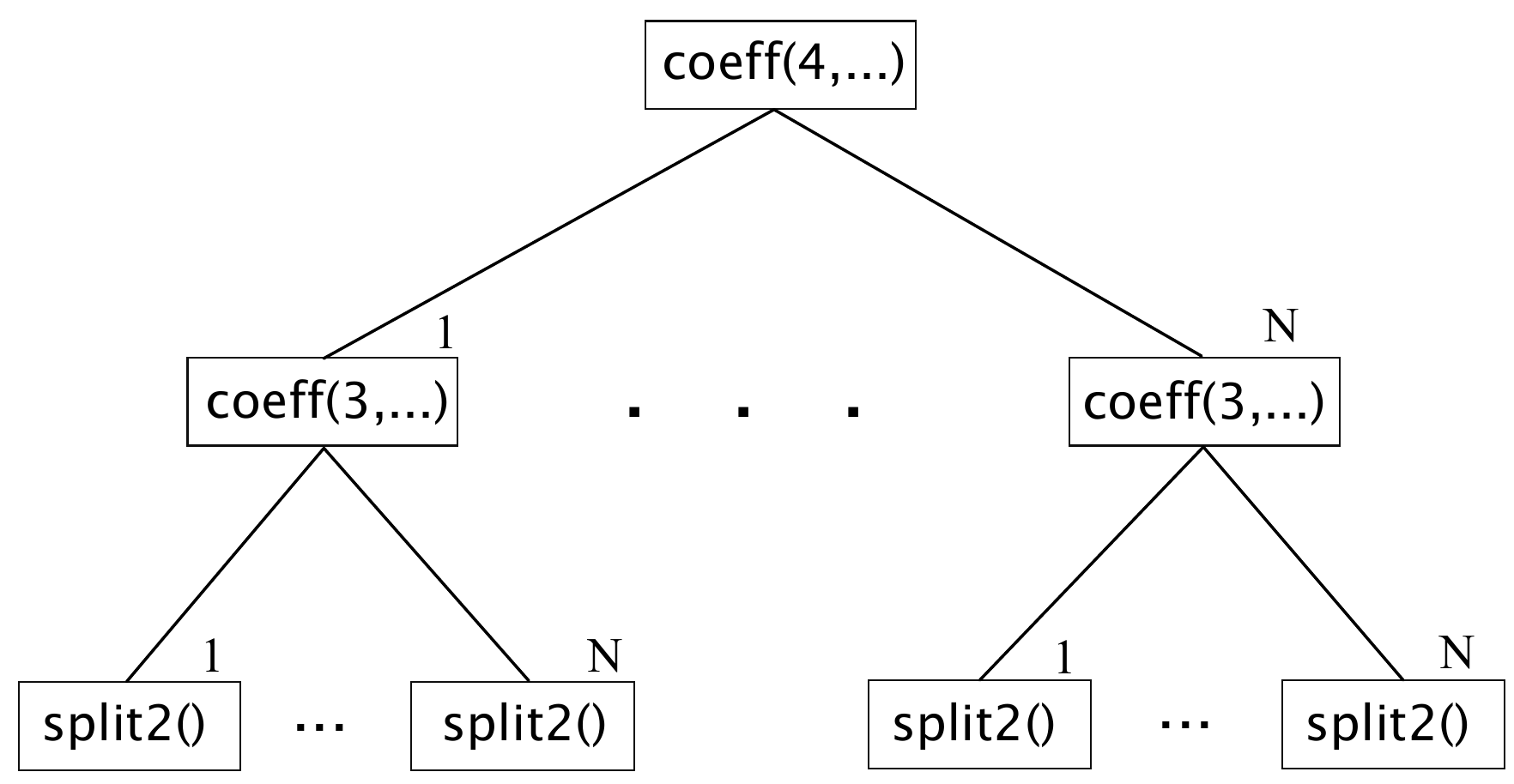}
\caption{Call tree of the algorithm for $k=4$.}
\end{center}
\end{figure}

\begin{rem} When $d$ is fixed, the running time depends on $p$ and the number of monomials in $f$.
\end{rem}  

\subsection{Implemetation details}
As was mentioned before, the algorithm can be parallelised in a very straight-forward way. As a target platform we choose a graphics card built around NVIDIA GeForce GTX 480 GPU\footnote{Graphics Processing Unit}. This GPU is capable of running up to 480 parallel execution threads. The GPU is programmed using an extension of C/C++ programming language called \emph{NVIDA CUDA\footnote{Compute Unified Device Architecture} C} \cite{cuda}. A typical CUDA program is executed on both the CPU and the GPU. The parts of the code that run on the graphics hardware are called \emph{kernels} and each kernel can be run $N$ times in parallel by $N$ different CUDA threads. The threads are organized in one-, two- or three-dimensional \emph{thread blocks} while blocks are organized in one- or two-dimensional \emph{grid}. The graphics hardware was origianlly created for the purpose of low-precision floating point operations it uveils its full computing power only when dealing with 32-bit floating point numbers. Therefore we can perform exact computations only on numbers within the range $-2^{23}\dots 2^{23}$ or equivalently $-8388608\dots 8388608$. That is because only 23 bits of 32 are used to store the significant digits (8-bits being reserved for the exponent and one remaining bit for the sign). This range is far too small to be useful due to the rapid growth of the numbers $\left[f^p\right]_0$.
Therefore we used a \emph{residue number system} (RNS) to represent large integers. Given a (fixed) set of pairwise coprime natural numbers $m_1,\dots,m_s$, the integer $x<M:=\prod m_i$ is represented by the system of residues
\begin{eqnarray*}
x_1:&=& x \mod m_1\\
&\vdots&\\
x_s:&=& x \mod m_s,
\end{eqnarray*}
$M$ is called the \emph{dynamic range} of the system. By the well-known \emph{Chinese Reminder Theorem} the number $x$ can be reconstructed from $(x_1,\dots,x_s;m_1,\dots,m_s)$. In practice, the conversion from an RNS representation of $x$ to a decimal representation is done with the \emph{Mixed Radix Conversion (MRC)} i.e. $x$ is represented in the form
\[x=a_1+a_2m_1+a_3m_1m_2+\dots+a_sm_1m_2\dots m_{s-1},\]
where $a_i$'s are called \emph{mixed radix digits}. Once $a_i$'s are known the decimal value can be computed easily. One classical algorithm for finding $a_i$'s from RNS representation of $x$ is given in \cite{Szabo}. 

In our implementation the program branches in several execution threads in the procedure COEFF on the top level of recursion. The threads are organized in $b$ blocks with $t$ threads per block. All the threads within the $k$-th block performs computation modulo prime number $p_k$. The $i$-th thread of $k$-th block computes $w_j\mod p_k$ for $j = i \mod t$. The arrangement of threads is illustrated in Figure 2.

\begin{figure}\label{threads}
\begin{center}
\includegraphics[width=6cm]{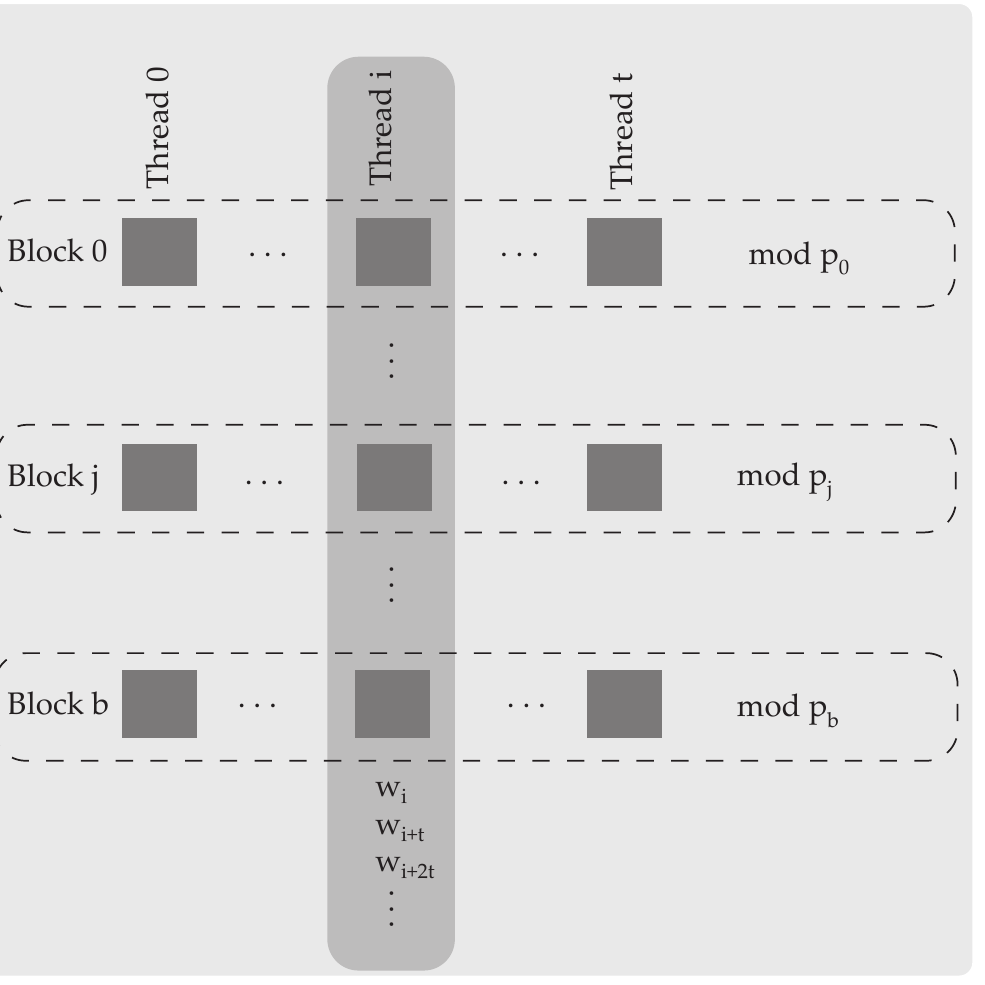}
\caption{Arrangement of threads.}
\end{center}
\end{figure}

\section{Application to toric Calabi-Yau varieties}
In the theory of toric Calabi-Yau varieties the constant terms of powers of Laurent polynomials occur in the following way. Consider a reflexive lattice polytope $\Delta\subset\RR^4$. Every vertex $v_i=(v_{i,1},\dots,v_{i,4})\in\ZZ^4$ of $\Delta$ corresponds to a monomial $\prod_{j=1}^4 X_j^{v_{i,j}}$. Then the Laurent polynomial corresponding to $\Delta$ is
\[f_\Delta=\sum_{i=1}^k \prod_{j=1}^4 X_j^{v_{i,j}}.\]
This polynomial defines a family of hypersurfaces $X_z:=\left\{1-z\cdot f_\Delta=0\right\}\subset(\CC^\times)^4$.
The function 
\begin{equation}
\Phi(z)=\sum_{i=0}^\infty \left[f_\Delta^i\right]_0 z^i
\end{equation} is the period of the holomorphic 3-form
\[\omega_z=\operatorname{Res}_{X_z}\left(\frac{1}{1-zf_\Delta(X)}\frac{dX_1}{X_1}\wedge\frac{dX_2}{X_2}\wedge\frac{dX_3}{X_3}\wedge\frac{dX_4}{X_4}\right)\]
on $X_z$ and is a solution of the Picard-Fuchs equation 
\begin{eqnarray}
L\Phi(z)&=&0\nonumber
\end{eqnarray}
where 
\begin{equation}
L=P_0(\theta)+zP_1(\theta)+\dots+z^kP_k(\theta),\label{dif-op}
\end{equation}
 $\theta:=z\frac{d}{dz}$ and $P_i$ are polynomials with integral coefficients
\begin{equation}\nonumber
P_i = \sum_{j=0}^{d_i} b_{ij}X^j.
\end{equation}
The numbers $(a_i)_{i\in\ZZ}$
\begin{equation}
a_i:=\begin{cases}\left[f^i\right]_0, & i\geq 0\\
      0, & \text{otherwise}
     \end{cases}
\end{equation}
 then satisfy the following recurrence relation
\begin{equation}
P_{0}(n)a_{n}+P_{1}(n-1)a_{n-1}+\dots+P_{k}(n-k)a_{n-k}=0 \label{eq-rec},\end{equation} for all $n\in\NN_0$. We say that the recurrence has length $k+1$ and degree $d:=\max\left\{\deg P_i\right\}$. Conversely, if a sequence satisfies (15) then its generating function is anihilated by the differential operator (13). 

Given $\Delta$ we want to find the corresponding differential operator $L$. To do this, we compute first $N$ coefficients of the constant terms series (12) of $f_\Delta$ and try to find a reccurence of the form (15). Assume there is one with lenght $k+1$ and degree $d$. Then it is determined by $(k+1)(d+1)$ coefficients $b_{ij}, \; i=0,\dots,k,\; j=0,\dots,d$ of $P_0,\dots,P_k$. We check if there is a solution to system of $N+1$ linear equations for the $(k+1)(d+1)$ unknowns $b_{ij}$
\begin{equation}\label{linsys}
\left\{
\begin{array}{l c c}
P_{0}(N)a_{N}+P_{1}(N-1)a_{N-1}+\dots+P_{k}(N-k)a_{N-k}&=&0\\
\vdots&\vdots&\vdots\\
P_{0}(2)a_{2}+P_{1}(1)a_{1}+P_{2}(0)a_{0}&=&0\\
P_{0}(1)a_{1}+P_{1}(0)a_{0}&=&0\\
P_{0}(0)a_{0}&=&0
\end{array}
\right.
\end{equation}
We have to keep this system overdetermined i.e. to have $N+1>(k+1)(d+1)$. If there is a solution which does not change when we add more equations, we hope that we found the correct operator. Since, a priori, we do not know the length and the degree of the reccurence we make the Ansatz (16) for all $k,n$ which are allowed by the requirement that (16) be overdetermined. Thus it is essential to be able to compute $[f^i]_0$ for as many $i$'s as possible in reasonable time.
 
This method was used by Batyrev and van Straten \cite{BvS} and has been polular ever since. More recently, Batyrev and Kreuzer used this approach in \cite{BK} to determine the  Picard-Fuchs operators for 
some new families of Calabi-Yau threefolds with Picard number 1. They succeeded in 28 cases out of 68. Using our algorithm were able to compute 4 more Calabi-Yau operators. Although we knew up to 300 constant terms, in many cases the recurrence did not show up. The operators and their corresponding polynomials are:\\

\noindent
\textbf{\#24}:

\noindent
\small

$\begin{array}{@{\hspace{0mm}}r@{\;}l@{\hspace{0mm}}}
f_{24}&=\frac{1}{T}+Y+{\frac {T}{X}}+{\frac {ZT}{X}}+{\frac {ZT}{XY}}+\frac{1}{Z}+{
\frac {X}{Z}}+{\frac {Y}{ZT}}+{\frac {X}{ZT}}+{\frac {XY}{ZT}}+{\frac 
{Y}{T}}+{\frac {T}{Y}}+{\frac {T}{XY}}+{\frac {Y}{X}}+\frac{1}{X}+{\frac 
{ZT}{Y}}+T\\
&+\frac{1}{Y}+X+{\frac {X}{T}}+{\frac {1}{ZT}}+Z+{\frac {Z}{Y}}
\end{array}$\\

$\begin{array}{@{\hspace{0mm}}r@{\;}l@{\hspace{0mm}}}
D_{24}&=97^2\,{\theta}^{4}+97 z\theta\left( -291-1300\,\theta-2018\,{\theta}^{2}+1727\,{\theta}^{3} \right)  \\
&+{z}^{2} \left(-2709792-10216234\,\theta-16174393\,{\theta}^{2}-13428812\,{\theta}^{3}-1652135\,{\theta}^{4}\right) \\
&+{z}^{3} \left(-138000348-443115594\,\theta-568639497\,{\theta}^{2}-364126194\,{\theta}^{3}-81753435\,{\theta}^{4}\right) \\
&+{z}^{4} \left(-3049275024-8869415520\,\theta-10006378570\,{\theta}^{2}-5423394464\,{\theta}^{3}-1175502862\,{\theta}^{4}\right)\\ 
&+{z}^{5} \left( -38537290992-103964102350\,\theta-106108023451\,{\theta}^{2}-50507429234\,{\theta}^{3}-9726250397\,{\theta}^{4}\right) \\
&+{z}^{6} \left( -308040167808-781527778884\,\theta-733053660150\,{\theta}^{2}-312374434824\,{\theta}^{3}-52762935894\,{\theta}^{4}\right) \\
&+{z}^{7} \left(-1619360309088-3901093356168\,\theta-3399527062044\,{\theta}^{2}-1313199235080\,{\theta}^{3}-195453433908\,{\theta}^{4}\right)\\
&-144 {z}^{8} \left( \theta+1 \right)  \left( 3432647479\,{\theta}^{3}+22487363787\,{\theta}^{2}+50808614711\,\theta+38959393614 \right) \\
&-432 z^9\left( \theta+2 \right)  \left( \theta+1 \right)  \left( 1903493629\,{\theta}^{2}+10262864555\,\theta+14314039440 \right)\\ 
&-438048 z^{10} \left( 1862987\,\theta+5992902 \right)  \left( \theta+3 \right)  \left( \theta+2 \right)  \left( \theta+1 \right) \\ 
&-368028363456 z^{11} \left( \theta+1 \right)  \left( \theta+2 \right)  \left( \theta+3 \right)  \left( \theta+4 \right)
\end{array}$\\

\noindent
\normalsize
\textbf{\#39}:

\noindent
\small

$\begin{array}{@{\hspace{0mm}}r@{\;}l@{\hspace{0mm}}}
f_{39}&={\frac {X}{YZ}}+{\frac {ZT}{XY}}+ZT+T+{\frac {Y}{X}}+{\frac {Z}{X}}+{\frac {Y}{XT}}+Y+{\frac {XY}{ZT}}+{\frac {X}{Z}}+{\frac {Y}
{T}}+Z+\frac{1}{T}+X+{\frac {X}{Y}}+{\frac {ZT}{Y}}+{\frac {X}{ZT}}\\
&+{\frac {Y}{ZT}}+\frac{1}{X}+\frac{1}{Y}+{\frac {ZT}{X}}+{\frac {T}{Y}
}+\frac{1}{Z}
\end{array}$\\

$\begin{array}{@{\hspace{0mm}}r@{\;}l@{\hspace{0mm}}}
D_{39}&=16\,{\theta}^{4}-4z\theta\, \left( 12+53\,\theta+82\,{\theta}^{2}+2\,{\theta}^{3} \right)\\
&+{z}^{2} (-5120-18308\,\theta-26199\,{\theta}^{2}-18410\,{\theta}^{3}-4895\,{\theta}^{4}) \\
&+{z}^{3} (-143808-430092\,\theta-497452\,{\theta}^{2}-272424\,{\theta}^{3}-60679\,{\theta}^{4})\\
&+{z}^{4}(-1478544-3987101\,\theta-4034628\,{\theta}^{2}-1870838\,{\theta}^{3}-344527\,{\theta}^{4})\\
&-{z}^{5}\left( \theta+1 \right)  \left( 1076509\,{\theta}^{3}+5847783\,{\theta}^{2}+11226106\,\theta+7492832 \right)\\
&-2 {z}^{6}\left( \theta+2 \right)  \left( \theta+1 \right)  \left( 944887\,{\theta}^{2}+4249317\,\theta+5045304 \right) \\
&-3328 {z}^{7}\left( 518\,\theta+1381 \right)  \left( \theta+3 \right)  \left( \theta+2 \right)  \left( \theta+1 \right)\\
&-621920 {z}^{8} \left( \theta+1 \right)  \left( \theta+2 \right)  \left( \theta+3 \right)  \left( \theta+4 \right)
\end{array}$\\

\noindent
\normalsize
\textbf{\#41}:

\noindent
\small

$\begin{array}{@{\hspace{0mm}}r@{\;}l@{\hspace{0mm}}}
f_{41}&=XT+YZ+ZT+{\frac {1}{ZT}}+T+{\frac {1}{XZ}}+{\frac {1}{YZ}}+{\frac {1}{XY}}+{\frac {1}{XT}}+XZT+{\frac {1}{YT}}+YZT+{\frac {1}{XYZ}}+{\frac {1}{XYT}}\\
&+{\frac {1}{XZT}}+{\frac {1}{YZT}}+XYZT+{\frac {1}{XYZT}}+Y+X+Z+\frac{1}{T}+\frac{1}{Z}+\frac{1}{X}+\frac{1}{Y}
\end{array}$\\

$\begin{array}{@{\hspace{0mm}}r@{\;}l@{\hspace{0mm}}}
D_{41}&=91^2\,{\theta}^{4}+91 z \theta\left( -273-1210\,\theta-1874\,{\theta}^{2}+782\,{\theta}^{3} \right) \\
&+{z}^{2} \left( -2649920-9962953\,\theta-15227939\,{\theta}^{2}-11622522\,{\theta}^{3}-2515785\,{\theta}^{4} \right) \\
&+{z}^{3} \left( -110445426-348819198\,\theta-432607868\,{\theta}^{2}-258678126\,{\theta}^{3}-59827597\,{\theta}^{4} \right)\\
&+{z}^{4} \left( -1915723890-5439732380\,\theta-5901995820\,{\theta}^{2}-2998881218\,{\theta}^{3}-612043042\,{\theta}^{4} \right) \\
&+{z}^{5} \left( -18479595006-48522700563\,\theta-47503242813\,{\theta}^{2}-21226829058\,{\theta}^{3}-3762840342\,{\theta}^{4} \right) \\
&+{z}^{6} \left(-110147546634-271941545379\,\theta-244753624741\,{\theta}^{2}-98210309094\,{\theta}^{3}-15265487382\,{\theta}^{4}\right) \\
&+{z}^{7}(-422269162452-991829482602\,\theta-831965057114\,{\theta}^{2}-304487632282\,{\theta}^{3}-42103272002\,{\theta}^{4})\\
&-2{z}^{8}\left( \theta+1 \right)  \left( 39253400626\,{\theta}^{3}+275108963001\,{\theta}^{2}+654332416678\,\theta+521254338620 \right) \\
&-{z}^{9} \left( \theta+2 \right)  \left( \theta+1 \right)  \left( 94987355417\,{\theta}^{2}+545340710193\,\theta+799002779040 \right) \\
&-1540 {z}^{10}\left( 43765159\,\theta+149264765 \right)  \left( \theta+3 \right)  \left( \theta+2 \right)  \left( \theta+1 \right)\\
&-2^2 3 5^2 7^2 11^2 11971 z^{11} \left( \theta+1 \right)  \left( \theta+2 \right)  \left( \theta+3 \right)  \left( \theta+4 \right) 
\end{array}$\\

\noindent
\normalsize
\textbf{\#38}:

\noindent
\small

$\begin{array}{@{\hspace{0mm}}r@{\;}l@{\hspace{0mm}}}
f_{38}&={\frac {X}{Y}}+{\frac {ZT}{Y}}+Z+X+T+{\frac {X}{Z}}+Y+{\frac {Z}{X}}+{\frac {YZ}{X}}+{\frac {1}{XZT}}+{\frac {Y}{XZT}}+{\frac {Y}{XT}}+{\frac {1}{ZT}}+\frac{1}{Z}+{\frac {1}{YZ}}+{\frac {T}{Y}}\\
&+{\frac {X}{YZ}}+{\frac {Y}{ZT}}+{\frac {1}{XT}}+\frac{1}{X}+\frac{1}{Y}
+{\frac {Y}{X}}+\frac{1}{T}
\end{array}$\\

\noindent
$\begin{array}{@{\hspace{0mm}}r@{\;}l@{\hspace{0mm}}}
D_{38}&= 102^2\theta^4-102 z \left( 204+911\,\theta+1414\,{\theta}^{2}+116\,{\theta}^{3} \right)  \\
&+{z}^{2}(-2663424-9947652\,\theta-14508941\,{\theta}^{2}-9892670\,{\theta}^{3}-2596259\,{\theta}^{4})\\
&+{z}^{3}(-67967496-206933112\,\theta-239004708\,{\theta}^{2}-125234088\,{\theta}^{3}-25685301\,{\theta}^{4}) \\
&+{z}^{4}(-598491604-1608054100\,\theta-1587508748\,{\theta}^{2}-687051032\,{\theta}^{3}-112357900\,{\theta}^{4})\\
&+{z}^{5}(-2495389956-6085656898\,\theta-5273754198\,{\theta}^{2}-1927713868\,{\theta}^{3}-254678692\,{\theta}^{4})\\
&+{z}^{6}(-5385015134-11995897911\,\theta-9101625228\,{\theta}^{2}-2758627602\,{\theta}^{3}-283337071\,{\theta}^{4}) \\
&+{z}^{7}(-5612134720-11209872916\,\theta-7075746650\,{\theta}^{2}-1555791344\,{\theta}^{3}-86504770\,{\theta}^{4}) \\
&+12 {z}^{8} \left( \theta+1 \right)  \left( 7613560\,{\theta}^{3}+27844427\,{\theta}^{2}-51849552\,\theta-134696600 \right) \\
&+{z}^{9} \left( \theta+2 \right)  \left( \theta+1 \right)  \left( 60585089\,{\theta}^{2}+495871401\,\theta+595115780 \right)   \\
&-600 z^{10}\left( 10279\,\theta-113205 \right)  \left( \theta+3 \right)  \left( \theta+2 \right)  \left( \theta+1 \right)  \\
&-6790000 z^{11} \left( \theta+1 \right)  \left( \theta+2 \right)  \left( \theta+3 \right)  \left( \theta+4 \right)
\end{array}$\\

\end{document}